\documentclass[12pt]{article}
\usepackage[papersize={8.5in,11in},textwidth=6in,bottom=1.25in]{geometry}

\usepackage{sectsty}
\sectionfont{\normalsize\bfseries}
\subsectionfont{\normalsize\sf}
\usepackage[mathcal]{euscript}
\usepackage[UKenglish]{babel}
\usepackage{bm}                     %AMS Packages
\usepackage{amsfonts}
\usepackage{amssymb}
\usepackage{amsmath}
\usepackage{amsthm}
\usepackage{graphicx}               %Graphics
\usepackage{natbib}                 %Citation style
\usepackage{colortbl}               %Tables
\usepackage{booktabs}
\usepackage{hyperref}
\usepackage{array,epsfig,fancyhdr}              %Graphics
\numberwithin{equation}{section}
\usepackage{multirow}

\usepackage[utf8]{inputenc}
\newtheorem{theorem}{Theorem}
\newtheorem{lemma}{Lemma}

\newtheorem{proposition}{Proposition}
\theoremstyle{definition}
\newtheorem{definition}{Definition}

\newtheorem{remark}{Remark}

\newcommand{\I}{\mathbb{I}}

\newcommand{\R}{\mathbb{R}}

\newcommand{\abs}[1]{\ensuremath{|{#1}|}}
\newcommand{\E}[1]{\ensuremath{\mathbb{E} \left[ {#1} \right]}}
\newcommand{\Ep}[1]{\ensuremath{\mathbb{E} \left( {#1} \right)}}
\newcommand{\V}{{\cal V}}
\newcommand{\EE}{\mathbb{E}}

\newcommand{\PP}{\mathbb{P}}

\begin{document}

%%%%%%%%%%%%%%%%%%%%%%%%%%%%%%%%%%%%%%%%%%%%%%%%%%%%%%%%%%%%%%%%%%%%%%%%%%%%%%%%%%%%%%%%%%%%%%%%%%%%%%%%%%%%%%%%%%%%%%%%%%%%
%%%%%%%%%%%%%%%%%%%%%%%%%%%%%%%%%%%%%%%%%%%%%%%%%%%%%%%%%%%%%%%%%%%%%%%%%%%%%%%%%%%%%%%%%%%%%%%%%%%%%%%%%%%%%%%%%%%%%%%%%%%%

\begin{center}
	\large \bf  Variable selection in functional data classification: a maxima hunting proposal\normalsize
	%\large \bf  On intrinsic variable selection in functional discrimination: 
	%two new proposals and some comparisons \normalsize
\end{center}
\normalsize

\begin{center}
	Jos\'e R. Berrendero\footnote[1]{joser.berrendero@uam.es}, Antonio Cuevas\footnote[2]{antonio.cuevas@uam.es}, Jos\'e L. Torrecilla\footnote[3]{joseluis.torrecilla@uam.es} \\
	Departamento de Matem\'aticas\\
	Universidad Aut\'onoma de Madrid, Spain
\end{center}

%%%%%%%%%%%%%%%%%%%%%%%%%%%%%%%%%%%%%%%%%%%%%%%%%%%%%%%%%%%%%%%%%%%%%%%%%%%%%%%%%%%%%%%%%%%%%%%%%%%%%%%%%%%%%%%%%%%%%%%%%%%%

\begin{abstract}
\footnotesize Variable selection is considered in the setting of supervised binary classification with functional data $\{X(t),\ t\in[0,1]\}$. By ``variable selection'' we mean  any dimension-reduction method which leads to replace the whole trajectory 
 $\{X(t),\ t\in[0,1]\}$, with a low-dimensional vector $(X(t_1),\ldots,X(t_d))$ still keeping a similar classification error.  
Our proposal for variable selection is based on the idea of selecting the local maxima $(t_1,\ldots,t_d)$ of the function ${\mathcal V}_X^2(t)={\mathcal V}^2(X(t),Y)$, where ${\mathcal V}$ denotes the  ``distance covariance'' association measure for random variables due to \citet{sze07}. This method provides a simple natural way to deal with the  relevance vs. redundancy trade-off which typically appears in variable selection. 
This paper includes 

(a) Some theoretical motivation: a result of consistent estimation for the maxima of ${\mathcal V}_X^2$ is shown. We also show different models for the underlying process $X(t)$ under which the relevant information is concentrated on the maxima of ${\mathcal V}_X^2$.  

(b) An extensive empirical study, including about 400 simulated models and real data examples, aimed at comparing our variable selection method with other standard proposals for dimension reduction.
\end{abstract}

\paragraph{Keywords:}
distance correlation, functional data analysis, supervised classification, variable selection.

%%%%%%%%%%%%%%%%%%%%%%%%%%%%%%%%%%%%%%%%%%%%%%%%%%%%%%%%%%%%%%%%%%%%%%%%%%%%%%%%%%%%%%%%%%%%%%%%%%%%%%%%%%%%%%%%%%%%%%%%%%%%

\def\thefigure{\arabic{figure}}
\def\thetable{\arabic{table}}

\fontsize{10.95}{14pt plus.8pt minus .6pt}\selectfont

\section{Introduction}\label{introd}

When dealing with functional data, the use of dimension reduction techniques arises as a most natural idea. Some of these techniques are based upon the use of general (linear) finite dimensional projections.  This is the case of functional principal component analysis (FPCA), see \citet{liy13}, although the so-called functional partial least squares (PLS) methodology is in general preferable when a response variable is involved; see \citet{del12} for a recent reference. Other common dimension reduction methods in the functional setting include sliced inverse regression (\citet{hsi09,jia14}) and additive models (\citet{zha13}). Also, the methods based on random projections could offer an interesting alternative.  See, e.g., \citet{cue14} for a short overview of dimension-reduction techniques together with additional references. 

\medskip

\it Some comments on the literature\rm. Our proposal here is concerned  with a different, more radical, approach to dimension reduction, given by the so-called \bf variable selection methods\rm. The aim of variable selection, when applied to functional data, is to replace every infinite dimensional observation $\{x(t),\ t\in[0,1]\}$, with a finite dimensional vector $(x(t_1),\ldots,x(t_d))$. 
The selection of the ``variables'' $t_1,\ldots,t_d$ should be a consequence of a trade-off between two mutually conflicting goals: representativeness and parsimony. In other words, we want to retain as much information as possible (thus selecting relevant variables) employing a small number of variables (thus avoiding redundancy). 

It is clear that variable selection has, at least, an advantage when compared with other dimension reduction methods (PCA, PLS...) based on general projections: the output of any variable selection method is always directly interpretable in terms of the original variables, provided that the required number $d$ of selected variables is not too large. As a matter of fact, variable selection is sometimes the main target itself in many cases where the focus is on model simplification.

We are especially interested in the ``intrinsic'' approaches to variable selection, in the sense that the final output should depend only on the data, not on any assumption on the underlying model (although the result should be interpretable in terms of the model).  There is a vast literature on these topics published by researchers in machine learning or by mathematical statisticians. The approaches and the terminology used in these two communities are not always alike. Thus, in machine learning, variable selection is often referred to as \it feature selection\rm. Also, the methods we have called ``intrinsic'' are often denoted as ``filter methods'' in machine learning.  
It is very common as well (especially in the setting of regression models) to use the terms ``sparse'' or ``sparsity'' to describe situations in which variable selection is the first natural aim; see e.g., \cite{ger10} and \cite{ros13}. It has been also argued in \cite{kne11} that the standard sparsity models are sometimes too restrictive so that it is advisable to combine them with other dimension reduction techniques. The ``relevant'' variables in a functional model are sometimes called ``impact points'' \citep{mck10} or ``most predictive design points'' \citep{fer10}. Also, the term ``choice of components'' has been  used by \cite{del12a}
as a synonym of variable selection. 

Let us finally mention, with no attempt of exhaustiveness in mind, that the recent literature in functional variable selection includes a version of the  classical lasso procedure \citep{zha14}, a study of consistency in the variable selection setup \citep{com12} and the use of inverse regression ideas in variable selection \citep{jialiu14}.
The monograph \cite{guy06} contains a complete survey on feature extraction (including selection) from the point of view of machine learning.  The overview paper by \cite{fan10} has a more statistical orientation.

%Despite this amount of bibliography, variable selection for functional data classification is rarely tackled in this context. Curves are not seen as functions but vectors of features, that can not be the original ones. That happens in \cite{gom09} (maybe the most related work with ours) where functional data are become into vectors of different components extracted from the curves and a new general feature selection method is applied to these vectors. 

\medskip

\it The functional classification problem\rm. In what follows we will focus on variable selection for the problem of supervised binary classification, with functional data.  While the statement and basic ideas behind the supervised classification (or discrimination) problem are widely known (see, e.g., \citet{dev96}), we need to briefly recall them for the sake of clarity and for notation purposes. Suppose  that an explanatory random variable $X$, taking values in a \it feature space\/ \rm ${\mathcal F}$, can be observed in the individuals of two populations $P_0$ and $P_1$. Let $Y$ denote a binary random variable, with values in $\{0,1\}$, indicating the membership to $P_0$ or $P_1$. 
On the basis of a data set ${\mathcal D}_n=((X_1,Y_1),\ldots,(X_n,Y_n))$ of $n$ independent observations drawn from $(X,Y)$, the supervised classification  problem aims at predicting the membership class $Y$ of a new observation for which only the variable $X$ is known.

A \it classifier\/ \rm or \it classification rule\/ \rm is just a measurable function
$g:{\mathcal F}\rightarrow \{0,1\}$. It is natural to assess the performance of a classifier by the corresponding \it classification error\/ \rm $L={\mathbb P}(g(X)\neq Y)$. It is well-known that the
classification error $L={\mathbb P}(g(X)\neq Y)$ is minimized by the so-called \it Bayes classifier\rm , $g^*(x)={\mathbb I}_{\{\eta(x)>1/2\}}$,
where $\eta(x)={\mathbb E}(Y|X=x)={\mathbb P}(Y=1|X=x)$. 
Since $g^*$ is in general unknown, it must be approximated, in different ways, by data-driven classifiers.

In our functional setting the feature space will be (unless otherwise stated)  ${\mathcal F}={\mathcal C}[0,1]$, the space of real continuous functions defined on $[0,1]$, endowed with the usual supremum norm.  Thus, our data will be of type $(X_1,Y_1),\ldots, (X_n,Y_n)$, where the $X_i$ are iid trajectories in ${\mathcal C}[0,1]$ drawn from a stochastic process $X=X(t)=X(\omega,t)$. When no confusion is possible, we will denote the whole process by $X$. When convenient, $X(t)$ will be denoted $X_t$.

Several functional classifiers have been considered in the literature (see, e.g., \citet{bai11b} for a survey). Among them, maybe the simplest one is the so-called $k$-nearest neighbors  rule ($k$-NN).  Additionally, we will also consider, as a simple standard choice, the classical linear Fisher's classifier (henceforth LDA), applied to the selected variables. 

\medskip

\it The purpose and contents of this paper\rm.

(a) In Section \ref{max-hunting} we propose a ``maxima hunting''
(MH)  method for variable selection. It is essentially based on the idea of selecting the local maxima $(t_1,\ldots,t_d)$ of the function ${\mathcal V}^2(t)={\mathcal V}^2(X(t),Y)$, where ${\mathcal V}^2$  denotes the 
``distance covariance''  association measure for random variables due to \citet{sze07}. An alternative version of the MH procedure can be obtained by replacing ${\mathcal V}^2(t)$ by the ``distance correlation'' ${\mathcal R}^2(t)$. 
See Section \ref{aux} for a short review of the definitions and properties of ${\mathcal V}^2$ and ${\mathcal R}^2$. 

  Some useful simplified versions for ${\mathcal V}^2$ are obtained in Th. \ref{expresiones} of Section \ref{max-hunting}, for the particular case where $Y$ is a binary variable. A result of consistent estimation (Th. \ref{th:consistency}) for the maxima of ${\mathcal V}^2$ is also proved in that section. 

(b) In Section \ref{motiv} we give several models (identified in terms of the conditional distributions $X(t)|Y=j$) in which the optimal classification rule depends only on a finite number of variables. We also show that in some of these models the variables to be selected coincide with the maxima of ${\mathcal V}^2$. These results provide a theoretical basis for the techniques of variable selection in functional classification models. Usually these techniques are considered from an exclusively algorithmic or computational point of view. It is therefore of some interest to motivate them in ``population terms'', by identifying some specific models where these techniques have full sense.  As pointed out by \cite{bia14}, \it ``Curiously, despite a huge research activity in this
area, few attempts have been made to connect the rich theory of stochastic
processes with functional data analysis''\rm. So the present paper can be seen as a contribution to partially fill this gap.

(c) An extensive simulation study, comparing our variable selection methods with other dimension reduction procedures (as well as with the ``baseline option'' of doing no variable selection at all) is included in Section \ref{sim}. Three real data examples are discussed in Section \ref{real}. Section \ref{conclusiones} includes some final conclusions as well as a ranking of all considered methods.

All the proofs are included in the Appendix.

\section{An auxiliary tool: the distance covariance}\label{aux}

The problem of finding appropriate association measures  between random variables (beyond the standard linear correlation coefficient) has received increasing attention in recent years; see for instance \citet{hal11}. We will use here the association measure proposed by \citet{sze07}, see also \citet{sze09}. It is called \it distance covariance\/ \rm (or \it distance correlation\/ \rm in the standardized version). It has a number of valuable properties: first, it can be used to define the association between two random variables $X$ and $Y$ of arbitrary (possibly different) dimensions; second, it characterizes independence in the sense that the distance covariance between $X$ and $Y$ is zero if and only if $X$ and $Y$ are independent; third, the distance correlation can be easily estimated in a natural plug-in way, with no need of smoothing or discretization.

\begin{definition}\label{def:dcov}
Given two random variables $X$ and $Y$ taking values in ${\mathbb R}^p$ and ${\mathbb R}^q$, respectively, let $\varphi_{X,Y}$, $\varphi_{X}$, $\varphi_{Y}$   be the characteristic functions
of $(X,Y)$, $X$ and $Y$, respectively. Assume that the components of $X$ and $Y$ have finite first-order moments. The distance covariance between $X$ and $Y$, is the non-negative number ${\cal V}(X,Y)$ defined by \begin{equation} \label{dcov}
{\cal V}^2(X,Y) = \int_ {\mathbb{R}^{p+q}} \mid \varphi_{X,Y}(u,v) -\varphi_X(u) \varphi_Y(v)\mid^2 w(u,v) du dv,
\end{equation}
with $w(u,v)= (c_p c_q \abs{u}_p^{1+p} \abs{v}_q^{1+q} )^{-1} $, where $c_d=\frac{\pi^{(1+d)/2}}{\Gamma((1+d)/2)}$ is half the surface area of the unit sphere in ${\mathbb R}^{d+1}$ and $|\cdot|_d$ stands for the Euclidean norm in ${\mathbb R}^d$. 
Finally, denoting ${\cal V}^2(X)={\cal V}^2(X,X)$, the (square) distance correlation is
defined by ${\cal R}^2(X,Y)=\frac{{\cal V}^2(X,Y)}{\sqrt{{\cal V}^2(X){\cal V}^2(Y)}}$ if ${\cal V}^2(X){\cal V}^2(Y)>0$, ${\cal R}^2(X,Y)$ $=0$ otherwise. 
\end{definition}

Note that these definitions make sense even if $X$ and $Y$ have different dimensions (i.e., $p\neq q$). 
 In addition, the association measure ${\cal V}^2(X,Y)$ can be consistently estimated through a relatively simple average of products calculated in terms of the mutual pairwise distances $|X_i-X_j|_p$ and $|Y_i-Y_j|_q$ between the sample values $X_i$ and the $Y_j$; see \citet[expression (2.8)]{sze09}. 
See also \cite{li12} for a different use of the correlation distance in variable selection. 

\section{Variable selection based on maxima hunting}\label{max-hunting}

Our proposal is based on a direct use of the distance covariance association measure. We just suggest to select the values of $t$ corresponding to local maxima of the distance-covariance function ${\cal V}^2(X_t,Y)$ or, alternatively, of the distance correlation function ${\cal R}^2(X_t,Y)$. This method has a sound intuitive basis as it provides a simple natural way to deal with the relevance vs. redundancy trade-off: the selected values must carry a large amount of information on $Y$, which takes into account the \it relevance \rm of the selected variables. In addition, the fact of considering local maxima automatically takes care of the \it redundancy   \rm problem, since the  highly relevant points close to the local maxima are automatically excluded from consideration.   This intuition is empirically confirmed by the results of Section \ref{sim}, where the practical performance of the maxima-hunting method is quite satisfactory. Figure \ref{fig:maxima} shows how the   fun
 ction ${\cal V}^2(X_t,Y)$ looks like in two different examples.

\begin{figure}[h!]\begin{center}
		
\includegraphics[scale=0.40]{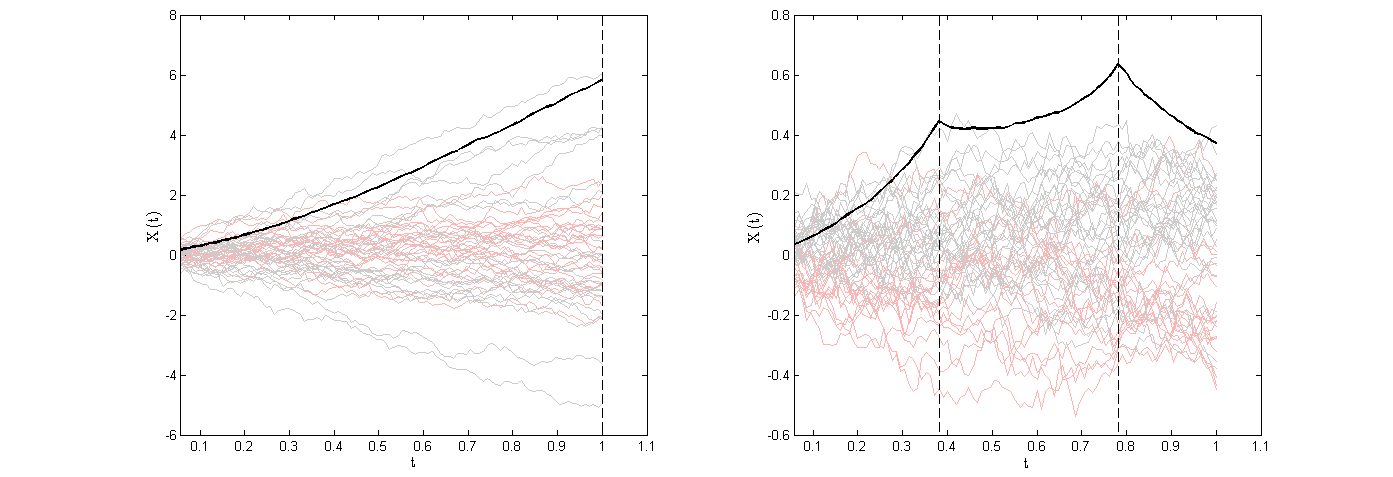}\par
\caption{\footnotesize Left: 50 trajectories of model in Proposition \ref{BvsBST}. Right: Logistic model L11 (explained in Subsection \ref{estruct}) with 50 Ornstein–Uhlenbeck trajectories. ${\cal V}^2(X_t,Y)$ (scaled) is in black and the relevant variables are marked by vertical dashed lines .}\label{fig:maxima}
	\end{center}
\end{figure}

The extreme flexibility of these association measures allows us to consider the case of
a multivariate response $Y$. So there is no conceptual restriction to apply the same ideas for multiple
classification or even to a regression problem. However, we will limit ourselves here to the important 
problem of binary classification. In this case we can derive simplified expressions for ${\cal V}^2(X,Y)$ which are particularly convenient in order to get empirical approximations. This is next shown. 
  
For the sake of generality, the results of this section will be obtained for the $d$-variate case, although in the rest of the paper we will use them just for $d$=1. Thus, throughout this subsection, $d$ will denote a natural number and $t$ will stand for  a vector $t=(t_1,\ldots, t_d)$ $\in [0,1]^d$. Also, for a given process $X$, we abbreviate $X(t)=(X(t_1),\ldots,X(t_d))$ by $X_t$ and $Z'$ will denote
an independent copy of a random variable $Z$. We write  $u^\top$ and $|u|_d$ to denote the transposed and the Euclidean norm of a vector $u\in\mathbb{R}^d$. Let $\eta(x)=\PP(Y=1|X=x)$  so that $Y|X \sim \mbox{Binomial}(1,\eta(X))$ where the symbol $\sim$ stands for ``is distributed as''. Observe that
$p=\PP(Y=1)=\mathbb{E}(\PP(Y=1|X))=\mathbb{E}(\eta(X))$.

Our variable selection methodology will heavily depend on the function $\V^2(X_t,Y)$ giving the distance covariance dependence measure between the marginal vector $X(t)=X_t$, for $t\in [0,1]^d$ and $d\in\mathbb{N}$, and the class variable $Y$. The following theorem gives three alternative expressions for this function. The third one will be particularly useful in what follows.

\begin{theorem}\label{expresiones} In the setting of the functional classification problem above stated, the function $\V^2(X_t,Y)$ defined in (\ref{dcov})
 can be alternatively calculated with the following expressions,
\begin{equation} \label{e1}
\hspace{0.25cm}(a) \hspace{1.5cm} {\cal V}^2(X_t,Y)=\frac{2}{c_d} \int_{\mathbb{R}^d} \frac{\abs{\zeta(u,t)}^2}{|u|_d^{d+1}}du,\hspace{3cm}
\end{equation}
where $\zeta(u,t)=\E{\left( \eta(X)-p\right)e^{iu^\top X_t}}$ and $c_d$ is given in Definition \ref{def:dcov}.
{\begin{align}\label{e2}
(b) \hspace{1.5cm} {\cal V}^2(X_t,Y)=& -2\E{(\eta(X)-p)(\eta(X')-p)|X_t-X'_t|_d}\nonumber \\
=&-2\E{(Y-p)(Y'-p)|X_t-X'_t|_d},
\end{align}}
where $(X^\prime, Y^\prime)$ denotes an independent copy of $(X,Y)$, respectively.
\begin{equation} \label{e3}
\hspace{0.5cm}(c) \hspace{1.5cm} {\cal V}^2(X_t,Y)=4p^2(1-p)^2 \left[ I_{01}(t) - \frac{I_{00}(t)+I_{11}(t)}{2}\right],
\end{equation}
where $I_{i j}(t)=\Ep{|X_t - X'_t|_d\, |\, Y=i,Y'=j}$.

\end{theorem}

\smallskip

In a training sample $\{(X_i,Y_i),\ i=1,\ldots,n\}$ denote by $X^{(0)}_1,\ldots, X^{(0)}_{n_0}$ and $X^{(1)}_1,\ldots, X^{(1)}_{n_1}$ the $X$-observations corresponding to values $Y_i=0$ and $Y_i=1$, respectively.  In this section, we use these data to obtain an  estimator of ${\cal V}^2(X_t,Y)$, which is uniformly consistent in $t$. As a consequence, we can estimate the local maxima of ${\cal V}^2(X_t,Y)$: using part (c) of Theorem \ref{expresiones}, a natural estimator for ${\cal V}^2(X_t,Y)$ is
\[
{\cal V}_n^2(X_t,Y)=4\hat p^2(1-\hat p)^2 \left[ \hat I_{01}(t) - \frac{\hat I_{00}(t)+\hat I_{11}(t)}{2}\right],
\]
where $\hat p=n_1 /(n_0 + n_1)$,  $\hat I_{rr}(t) = \frac{2}{n_r(n_r-1)} \sum_{i<j} |X^{(r)}_i(t) - X^{(r)}_j(t)|_d,$
 for $r=0,1$, and $\hat I_{01}(t) = \frac{1}{n_0n_1} \sum_{i=1}^{n_0} \sum_{j=1}^{n_1} |X^{(0)}_i(t) - X^{(1)}_j(t)|_d.$
The uniform strong consistency of ${\cal V}_n^2(X_t,Y)$ is established in Theorem \ref{th:consistency} below. 

\begin{theorem}
\label{th:consistency}
Let $X=X_t$, with $t\in[0,1]^d$, be a process with continuous trajectories almost surely such that $\mathbb{E}( \|X\|_\infty\log^+\|X\|_\infty )< \infty$. Then, ${\cal V}_n^2(X_t,Y)$ is continuous in $t$ and
\[
\underset{t\in [0,1]^d}{\sup}|{\cal V}_n^2(X_t,Y)-{\cal V}^2(X_t,Y)|
\to 0 \ \ \mbox{a.s.,\ as } n\to\infty.
\]
Hence, if we assume that ${\cal V}^2(X_t,Y)$ has exactly $m$ local maxima at $t_1,\cdots,t_m$, then ${\cal V}_n^2(X_t,Y)$ has also eventually at least
$m$ maxima at $t_{1n},\cdots,t_{mn}$ with $t_{jn}\to t_j$, as $n\to\infty$, a.s., for $j=1,\ldots,m$.
\end{theorem}

\section{Some theoretical, model-oriented motivation for variable selection and maxima-hunting}\label{motiv}
The variable selection methods we are considering here for the binary functional classification problem are aimed at selecting \it a finite number of variables\rm.  One might think that this is a ``too coarse'' approach for functional data. Nevertheless, we provide here some  theoretical motivation by showing that, in some relevant cases, variable selection is ``the best we can do'' in the sense that, in some relevant models, the Bayes rule (i.e., the optimal classifier) has an expression of type $g^*(X)=h(X(t_1),\cdots,X(t_d))$, so that it  depends only on a finite (typically small) number of variables. In fact, in many situations, a proper variable selection leads to an improvement in efficiency (with respect to the baseline option of using the full sample curves), due to the gains associated with a smaller noise level.

 The distribution of
$X(t)|Y=i$, will be denoted by $\mu_i$ for $i=0,1$. In all the examples below the considered processes are Gaussian, i.e., for all $t_1,\ldots,t_m\in[0,1]$, with $m\in{\mathbb N}$, the finite-dimensional marginal $(X(t_1),\ldots,X(t_m))|Y=i$ has a normal distribution in ${\mathbb R}^m$ for $i=0,1$. Many considered models have non-smooth, Brownian-like trajectories. These models play a very relevant role in statistical applications, in particular to the classification problem; see, e.g., \citet{lin09}.

Let us now recall some basic notions and results to be used throughout (see, e.g., \citet[ch. 4]{ath06}, for details): $\mu_0$ is said to be \it absolutely continuous with respect to $\mu_1$\/ \rm (which is denoted by $\mu_0 \ll\mu_1$) if and only if $\mu_1(A)=0$ entails $\mu_0(A)=0$, $A$ being a Borel set in ${\mathcal C}[0,1]$. Two probability measures $\mu_0$ and $\mu_1$ are said to be \it equivalent\/ \rm if  $\mu_0 \ll\mu_1$ and $\mu_1 \ll\mu_0$; they are \it mutually singular\/ \rm when there exists a Borelian set $A$ such that $\mu_1(A)=0$ and $\mu_0(A)=1$.
The so-called \it Hajek-Feldman dichotomy \rm (see \citet{fel58}) states that  if $\mu_0$ and $\mu_1$ are Gaussian, then they are either equivalent or mutually singular.
The \it Radon-Nikodym Theorem\/ \rm establishes that $\mu_1\ll \mu_0$ if and only if there exists a measurable function $f$ such that
$\mu_1(A)=\int_Af d\mu_0$ for all Borel set $A$. The function $f$ (which is unique $\mu_0$-almost surely) is called \it Radon-Nikodym derivative of $\mu_1$ which respect to $\mu_0$\rm. It is usually represented by $f=\frac{d\mu_1}{d\mu_0}$.

Finally, in order to obtain the results in this section we need to recall (see \citet[Th. 1]{bai11a}) that
\begin{equation}
\label{eqBayesRN}
\eta(x)=\left[\frac{1-p}{p}\frac{d\mu_0}{d\mu_1}(x)+1 \right]^{-1},\ \ \mbox{for}\ x\in {\mathcal S},
\end{equation}
where ${\mathcal S}$ is the common support of $\mu_0$ and $\mu_1$, and $p=\PP(Y=1)$. This equation
provides the expression for the optimal rule $g^*(x)=\mathbb{I}_{\{\eta(x)>1/2\}}$ in some important cases where the
Radon-Nikodym derivative is explicitly known.

\medskip

\it Some examples\rm. \noindent  Two non-trivial situations in which the Radon-Nikodym derivatives can be explicitly calculated are those problems where $\mu_0$ is the standard Brownian motion $B(t)$, and $\mu_1$ corresponds to $B(t)$ plus a stochastic or a linear trend. In both cases the Bayes rule $g^*$ turns out to depend just on one value of $t$.  To be more precise, it has the form $g^*(X)=h(X(1))$. This is formally stated in the following results. Proofs can be found in the Appendix.

\begin{proposition}\label{BvsBST}
Let us assume that $\mu_0$ is the distribution of a standard Brownian motion $B(t),\ t\in[0,1]$ and $\mu_1$
is the distribution of $B(t)+\theta t$, where $\theta$ is a random variable with distribution $N(0,1)$, independent from $B$. Then, the Bayes rule is given by 
$g^*(x)={\mathbb I}_{\left\{x_1^2 > 4\log\left( \frac{\sqrt{2}(1-p)}{p} \right)\right\}}(x),\ \ \mbox{for all}\ x\in{\mathcal C}[0,1]$.
\end{proposition}

As a particular case, when the prior probabilities of both groups are equal, $p=1/2$, we get $g^*(x)=1$ if and only if
$
|x_1| > 2\sqrt{\log\sqrt{2}} \approx 1.77.
$

\begin{proposition}\label{BvsBLT}
Let us assume that $\mu_0$ is the distribution of a standard Brownian motion $B(t),\ t\in[0,1]$ and $\mu_1$
is the distribution of $B(t)+c t$, where $c\neq 0$ is a constant. Then, for $x\in{\mathcal C}[0,1]$ the Bayes rule is given by
$g^*(x)={\mathbb I}_{\left\{x_1 > \frac{c}{2} - \frac{1}{c}\log\left(\frac{p}{1-p}\right)\right\}}(x)$, if $c>0$, and 
$g^*(x)={\mathbb I}_{\left\{x_1 < \frac{c}{2} - \frac{1}{c}\log\left(\frac{p}{1-p}\right)\right\}}(x)$, if $c<0$.

\end{proposition}

Before presenting our third example we need some additional notation. Let us now define the countable family of \it Haar functions\rm, $
\varphi_{m,k}=\sqrt{2^{m-1}}
\left[ \I_{\left( \frac{2k-2}{2^m},\frac{2k-1}{2^m}\right)} \right.$ $\left. - 
 \I_{\left( \frac{2k-1}{2^m},\frac{2k}{2^m}\right)}\right],
$
for $m,k\in{\mathbb N}$, $1\leq k\leq 2^{m-1}$.
The family $\{\varphi_{m,k}\}$ is known to be an orthonormal basis in $L^2[0,1]$. Moreover, define the ``peak'' functions $\Phi_{m,k}$
by
\begin{equation}
\Phi_{m,k}(t)=\int_0^t\varphi_{m,k}(s)ds.\label{peak}
\end{equation}
We want to use these peak functions to define the trend of the $\mu_1$ distribution in another model of type ``Brownian versus Brownian plus trend''. In this case the Bayes rule depends just on three points.
\begin{proposition}\label{BvsBTri}
Let us assume that $\mu_0$ is the distribution of a standard Brownian motion $B(t),\ t\in[0,1]$ and $\mu_1$
is the distribution of $B(t)+\Phi_{m,k}(t)$, where $\Phi_{m,k}$ is one of the peak functions defined above. Then, for $x\in{\mathcal C}[0,1]$ the regression function $\eta(x)={\mathbb E}(Y|X=x)$ is
\begin{align}
\label{etaphi}
\eta(x)=\left\{
\frac{1-p}{p}\exp \left( \frac{1}{2} - 2^{\frac{m-1}{2}} 
\left[\left( x_{\frac{2k-1}{2^m}}-x_{\frac{2k-2}{2^m}}\right)+
\left( x_{\frac{2k-1}{2^m}}-x_{\frac{2k}{2^m}}\right) \right]   
 \right)+1 \right\}^{-1}
\end{align}
and the Bayes rule
$ g^*(x)={\mathbb I}_{\{\eta(x)>1/2\}}$ fulfils $g^*(x)=1$ if and only if
\begin{align}
\label{etarule}
 \left( x_{\frac{2k-1}{2^m}}-x_{\frac{2k-2}{2^m}}\right) +
\left( x_{\frac{2k-1}{2^m}}-x_{\frac{2k}{2^m}}\right) 
>
\frac{1}{\sqrt{2^{m+1}}} -
\frac{1}{\sqrt{2^{m-1}}} \log\left(\frac{p}{1-p}\right).
\end{align}
\end{proposition}
Let us recall that, according to Cameron-Martin Theorem (see \citet[p. 24]{mor10}), in order to get the equivalence of $\mu_1$ and $\mu_0$ the trend function is required to belong to the Dirichlet space ${\mathcal D}[0,1]$ of real functions $F$ defined in $[0,1]$ which have a derivative $F^\prime$ in $L^2[0,1]$ such that $F(t)=\int_0^tF^\prime(s)ds$. 
It can be seen (\citet[p. 28]{mor10}) that $\{\Phi_{m,k}\}$ is an orthonormal basis for
${\mathcal D}[0,1]$.
\begin{remark} Analogous calculations can be performed (still obtaining explicit expressions for the Bayes rule of type $g^*(x)=g(x(t_1),\ldots,x(t_d))$), using a rescaled Brownian motion $\sigma B(t)$ or the Brownian Bridge instead of $B(t)$, or a piecewise linear trend instead of these. Likewise, other models could be obtained by linear combinations in the trend functions or by finite mixtures of other simpler models. Many of them have been included in the simulation study of Section \ref{sim}.
\end{remark}
Next, we will provide some theoretical support for the maxima-hunting method, by showing that in some specific useful models the optimal classification rule depends on  the maxima of the distance covariance function ${\cal V}^2(X_t,Y)$, although in some particular examples, other points (closely linked to the maxima) are also relevant.
\begin{proposition}\label{maximo-unico}
Under the models assumed in Propositions \ref{BvsBST} and \ref{BvsBLT}, the corresponding distance covariance functions $\V^2(X_t,Y)$ have both a unique relative maximum at the point $t=1$.
\end{proposition}

\begin{remark}\label{rem:maximo}
Other similar results could be obtained for the model considered in Proposition \ref{BvsBTri} as well as for the Brownian bridge vs. Brownian motion model. 
\end{remark}

The model considered in Proposition \ref{BvsBST} provides a clear example of  the advantages of using the distance covariance measure $\V^2(X_t,Y)$ rather than the ordinary covariance 
$Cov^2(X_t,Y)$ in the maxima-hunting procedure. Indeed, note that in this case, $Cov^2(X_t,Y) = p^2(1-p)^2({\mathbb E}(X(t)|Y=0)-{\mathbb E}(X(t)|Y=1))^2 = 0,$ for all $t\in[0,1]$, 
so that the ordinary covariance is useless to detect any difference between 
the values of $t$.

\section{A simulation study}\label{sim}

We describe here in detail the methods under study and the models to be considered together with a summary of the results. The full outputs can be found in \url{www.uam.es/antonio.cuevas/exp/outputs.xlsx}.

\subsection{The variable selection methods under study. Criteria for comparisons}\label{comp}
These are the methods, and their corresponding  notations as they appear in the tables and figures below.

 1. {\bf Maxima-hunting}.  The functional data $x(t),$  $t\in[0,1]$ are discretized to $(x(t_1),\ldots,$ $x(t_N))$, so a non-trivial practical problem is to decide which points in the grid are the local maxima: a point $t_i$ is declared to be a local maximum when it is the highest local maximum on the sub-grid $\{t_j\}$, $j=i-h\ldots,i+h$. The proper choice of $h$ depends on the nature and discretization pattern of the data at hand. Thus, $h$ could be considered as a smoothing parameter to be selected in an approximately optimal way. In our experiments $h$ is chosen by a validation step explained in next section. 

Then, we sort the maxima $t_i$  by \bf relevance \rm (the value of the function at $t_i$). This seems to be the natural order and it produces better results than other simple sorting strategies. We denote these maxima-hunting methods by \textbf{MHR} and \textbf{MHV} depending on the use of ${\cal R}^2$ or ${\cal V}^2$.

  2. \bf Univariate $t$-ranking method\/\rm, denoted by \textbf{T}, is frequently used when selecting relevant variables (see e.g. the review by \citet{fan10}). It is based on the simple idea of selecting the variables $X_t$ with highest Student's $t$ two-sample scores
$T(X_t)=|\bar{X}_{1t}-\bar{X}_{0t}|/\sqrt{s^2_{1t}/n_1+s^2_{0t}/n_0}$.

 3. {\bf mRMR}. The minimum Redundancy Maximum Relevance algorithm, proposed in \citet{din05} and \citet{pen05},  is a relevant intrinsic variable selection method; see \citet{ber15} for a recent contribution. It aims at maximizing the relevance of the selected variables avoiding an excess of redundancy what seems particularly suitable for functional data. Denoting the set of selected variables by $S$, the variables are sequentially incorporated to $S$ with the criterion of maximizing the difference $Relevance(S)-Redundancy(S)$ (or alternatively the quotient $Relevance(S)/Redundancy(S)$). Two ways of measuring relevance and redundancy have been proposed: first, we can use the Fisher statistic for relevance and the standard correlation for redundancy. Second, a three-fold discretized version of the so-called \textit{Mutual Information} measure for both relevance and redundancy (see \citet[equation (1)]{din05}).

In principle these two approaches are intended for continuous and discrete variables respectively. However, \citet{din05} report a good performance for the second one even in the continuous case. We have considered mRMR as a natural competitor for our maxima-hunting approximation. We have computed both Fisher-Correlation and Mutual Information approaches with both difference and quotient criteria. For the sake of clarity we only show here the results of \textbf{FCQ} (Fisher Correlation Quotient) and \textbf{MID} (Mutual Information Difference) which outperform on average their corresponding counterparts. %The results for FCD and MIQ can be seen in \url{www.uam.es/antonio.cuevas/exp/outputs.xlsx}.

  4. {\bf PLS}. According to the available results (\citet{pre07,del12}) PLS is the ```method of choice'' for dimension reduction in functional classification. Note however that PLS is not a variable selection procedure; in particular it lacks the interpretability of variable selection. In some sense, the motivation for including PLS is to check how much do we lose by restricting ourselves to variable selection methods, instead of considering other more general linear projections procedures (as PLS) for dimension reduction.

 5. {\bf Base}. The $k$-NN classifier is applied to the entire curves. The Base  performance can be seen as a reference to assess the usefulness of dimension reduction methods. Somewhat surprisingly, Base is often outperformed. Note that the Base method cannot be implemented with  LDA 
since this classifier typically fails with infinite or high-dimensional data; see, e.g. \citet[Section 6.1]{cue14}, for some insights and references.

The \bf classifiers \rm used in all cases are either $k$-NN, based on the Euclidean distance or  LDA (applied to the selected variables). Similar comparisons could be done with other classifiers, since the considered methods do not depend on the classifier. For comparing the different methods we use the natural accuracy measure, defined by the percentage of correct classification.

\subsection{The structure of the simulation study}\label{estruct}
Our simulation study consists of 400 experiments, aimed at comparing the practical performances of several intrinsic variable selection methods described in the previous subsection. These experiments  are obtained by considering 100 different underlying models and 4 sample sizes, where by ``model'' we mean either,
\begin{itemize}
\item[(M1)] a pair of distributions for $X|Y=0$ and $X|Y=1$ (corresponding to $P_0$ and $P_1$, respectively); in all cases, we take $p={\mathbb P}(Y=1)=1/2$.
\item[(M2)] The marginal distribution of $X$ plus $\eta(x)={\mathbb P}(Y=1|X=x)$. 
\end{itemize}
Models vary in difficulty and number of relevant variables. In all the considered models the optimal Bayes rule turns out to depend on a finite number of relevant variables, see Section 3. The processes involved  include also different levels of smoothing.  The full list of considered models is available in the 
Supplementary Material document. All of them belong to one of the following classes:

1. \bf Gaussian models\rm: they are denoted $G1, G1b,\ldots, G8$. All of them are generated according to the general pattern (M1). In all cases the distributions of $X(t)|Y=i$ are 
chosen among one of the following types: first, the \bf standard Brownian Motion\rm, $B$, in $[0,1]$. Second,
\bf Brownian Motion, $BT$, with a trend \rm $m(t)$, i.e., $BT(t)$ $=B(t)+m(t)$ (we have considered several choices for $m(t)$).
Third, the \bf Brownian bridge\rm: $BB(t)=B(t)-tB(1)$. Our fourth class of Gaussian processes is the \bf Ornstein–Uhlenbeck process\rm, with a covariance function of type $\gamma(s,t)=a\exp(-b|s-t|)$ and zero mean ($OU$) or different mean functions $m(t)$ ($OUt$). Finally smoother processes have been also computed by convolving Brownian trajectories with Gaussian kernels. We have considered two levels of smoothing denoted by sB and ssB.

 2. \bf Logistic models\rm : they are defined through the general pattern (M2): the process $X=X(t)$ follows one of the above mentioned distributions and $Y\sim\mbox{Binom}(1,\eta(X))$ with 
$\eta(x)=(1+e^{-\Psi(x(t_1),\cdots,x(t_d))})^{-1}$,
a function of the relevant variables $x(t_1),\cdots,x(t_d)$. We have considered 15 versions of this model and a few variants, denoted $L1, L2$, $L3, L3b, \ldots, L15$. They correspond to 
different choices for the link function $\Psi$ (most of them linear or polynomial) and for the distribution of $X$.  For example, in the models L2 and L8 we have $\Psi(x)=10x_{30}+10x_{70}$ and $\Psi(x)=10x_{50}^4+50x_{80}^3+20x_{30}^2$, respectively. 

 3.  \bf Mixtures\rm: they are obtained by combining (via mixtures) in several ways the above mentioned Gaussian distributions assumed for $X|Y=0$ and $X|Y=1$. These models are denoted M1, ..., M11 in the output tables.

For each model, all the variable selection methods (as well as PLS) are checked for sample sizes $n=30$, 50, 100, 200. So we get $100\times 4=400$ experiments.  

All the functional simulated data  are \bf discretized \rm to $(x(t_1), \ldots, x(t_{100}))$, where $t_i$ are equispaced points in $[0,1]$. In fact (to avoid the degeneracy $x(t_0)=0$ in the Brownian-like models) we take $t_1=6/105$. Similarly, for the case of the Brownian bridge, we truncate as well at the end of the interval.

The involved parameters  are: the number $k$ of nearest neighbors in the $k$-NN classifier, the dimension of the reduced space (number of variables or PLS components) and the smoothing parameter $h$ in maxima-hunting methods. These are set by standard data-based validation procedures. Parameter validation can be carried out mainly through a validation set or by cross-validation on the training set (see e.g. \cite{guy06}).  In the case of the simulation study, validation and test samples of size 200 are randomly generated. In the real data sets we proceed by cross-validation.

\subsection{A few numerical outputs from the simulations}\label{outputs}

We have selected (with no particular criterion in mind) a sampling of just a few examples among the 400 experiments. The complete simulation outputs can be downloaded from 
\url{www.uam.es/antonio.cuevas/exp/outputs.xlsx}. 
 Table 1 provides the performance (averaged on 200 runs) measured in terms of classification accuracy (percentages of correct classification). Models are presented in rows and methods in columns.
The marked outputs correspond to the winner and second best method in each row.

\begin{table}
	\caption{\footnotesize Average correct classification outputs, over 200 runs, with $n=50$.  } 
	\begin{center}
		\begin{footnotesize}
			\begin{tabular}{lccccccc}\hline\noalign{\smallskip}
				\multicolumn{8}{c}{\rm \bf $k$-NN outputs}\\      
				Models & FCQ & MID & T & PLS & MHR & MHV & Base\\
				\hline\noalign{\smallskip}
				L2\_OUt & 82.47 & 82.11 & 81.68 & \framebox{ 83.27} & 83.22 & \framebox{ 83.23} & 82.60\\
				L6\_OU & 88.41 & 89.81 & 86.19 & \framebox{ 90.93} & 90.75 & \framebox{ 90.83} & 90.56\\
				L10\_B & 81.09 & 85.02 & 81.13 & 85.90 & \framebox{ 87.27} & \framebox{ 87.42} & 85.46\\
				L11\_ssB & 82.31 & 80.85 & 82.28 & 78.81 & \framebox{ 83.10} & \framebox{ 82.81} & 79.89\\
				L12\_sB & 77.24 & 75.83 & \framebox{ 77.41} & 74.92 & \framebox{ 78.57} & 76.62 & 74.78\\
				G1 & 65.86 & 70.70 & 65.57 & 66.95 & \framebox{ 71.59} & \framebox{ 71.80} & 70.10\\
				G3 & 63.09 & 73.39 & 60.57 & 60.56 & \framebox{ 77.47} & \framebox{ 77.06} & 65.26\\
				G6 & 84.27 & 91.95 & 84.14 & \framebox{ 93.67} & 93.38 & \framebox{ 93.71} & 92.19\\
				M2 & 70.77 & 69.82 & 69.16 & \framebox{ 78.16} & 74.76 & \framebox{ 75.68} & 71.14\\
				M6 & 81.15 & 83.08 & 79.73 & \framebox{ 83.47} & \framebox{ 83.32} & 83.35 & 80.99\\
				M10 & 64.93 & 68.33 & 64.58 & 68.25 & \framebox{ 70.66} & \framebox{ 70.94} & 68.95\\
				\noalign{\smallskip}\hline  \noalign{\smallskip}\noalign{\smallskip}
				\multicolumn{8}{c}{\rm \bf LDA outputs}\\    
				Models & FCQ & MID & T & PLS & MHR & MHV & Base\\
				\hline\noalign{\smallskip}
				L2\_OUt & 79.80 & 78.95 & 78.23 & 80.07 & \framebox{ 80.24} & \framebox{ 80.14} & -\\
				L6\_OU & 87.79 & 88.91 & 84.46 & \framebox{ 91.01} & \framebox{ 89.44} & 89.35 & -\\
				L10\_B & 75.97 & 75.44 & 76.04 & 77.60 & \framebox{ 77.63 }& \framebox{ 77.76} & -\\
				L11\_ssB & 80.95 & 80.09 & 80.81 & 79.39 & \framebox{ 81.88} & \framebox{ 81.63} & -\\
				L12\_sB & 76.39 & 75.20 & \framebox{ 76.40} & 75.02 & \framebox{ 77.38} & 75.96 & -\\
				G1 & 51.27 & 51.24 & 51.20 & 51.44 & \framebox{51.55} &\framebox{ 51.70} & -\\
				G3 & 51.09 & 52.26 & 50.96 & 50.35 & \framebox{52.95} & \framebox{52.69} & -\\
				G6 & 87.72 & 95.28 & 87.80 & \framebox{ 97.77} & 96.54 & \framebox{ 96.85} & -\\
				M2 & 67.44 & 76.51 & 66.81 & \framebox{ 84.38} & 82.24 & \framebox{ 83.06} & -\\
				M6 & 79.99 & 79.92 & 79.63 & \framebox{ 81.39} & 81.08 & \framebox{ 81.38} & -\\
				M10 & 60.03 & 65.61 & 59.24 & \framebox{ 67.49} & 67.25 & \framebox{ 67.99} & \\
				\noalign{\smallskip}\hline
			\end{tabular}
		\end{footnotesize}
	\end{center}
\end{table}

The outputs of Table 1 are more or less representative of the overall conclusions of the entire study. For instance, MHR appears as the overall winner on average with a slight advantage. PLS and the maxima-hunting methods (MHR and MHV) obtain similar scores and clearly outperform the other benchmark methods.  Note that they also beat (often very clearly) the Base method in almost all cases  using just a few variables. This shows that dimension reduction is, in fact, ``mandatory'' in many cases.  Regarding the comparison of $k$-NN and LDA in the second stage (after dimension reduction) the results show a slight advantage for $k$-NN (on average). The complete failure of LDA in models G1 and G3 was to be expected since in these cases the mean functions are identical in both populations.  In terms of number of variables, when $k$-NN is used, MHR and MHV need less variables to achieve better results than the rest of variable selection methods. When LDA is used, the
  number of required variables is quite similar in all methods; see the Supplementary Material, Section S4.

\section{Real data examples}\label{real}
 
We have chosen three examples due to their popularity in FDA. There are many references on these datasets so we will just give  brief descriptions of them; additional details can be found in the Supplementary Material document. Figure \ref{fig:reales} shows the trajectories $X(t)$ and mean functions for each set and each class.

\begin{figure}[h!]\begin{center}
		
		\includegraphics[scale=0.5]{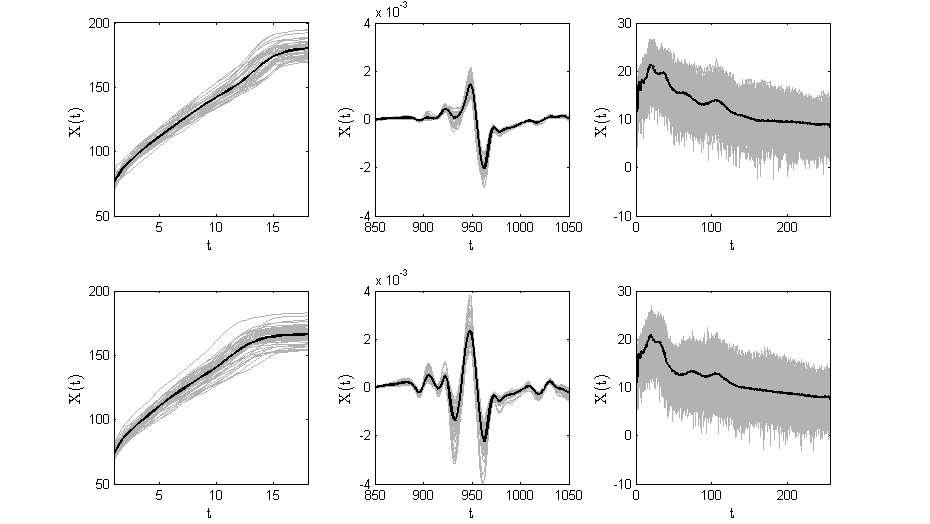}\par
		\caption{\footnotesize Data trajectories and mean functions from class 0 (first row) and class 1 (second row). Columns correspond to growth, Tecator and phoneme data from left to right.}\label{fig:reales}
	\end{center}
\end{figure}

{\it Berkeley Growth Data.}
The heights of 54 girls and 39 boys measured at 31 non equidistant time points. See, e.g., \citet{ram05}.

{\it Tecator.} 215 near-infrared absorbance spectra (100 grid points each) of finely chopped meat, obtained using a Tecator Infratec Food \& Feed Analyzer. The sample is separated in two classes according to
the fat content (smaller or larger than 20\%). 
Tecator curves are  often used in a differentiated version. We use here the second derivatives. See \citet{fer06} for details.

{\it Phoneme.}
As in \citet{del12a} we use the ``binary'' version of these data corresponding to log-periodograms constructed from 32 ms long
recordings of males pronouncing the phonemes  ``aa'' and ``ao''. The sample size is $n=1717$ ($695$ from ``aa'' and  $1022$ from ``ao''). Each curve was observed at 256 equispaced points.

In the comparisons with real data sets we have incorporated the method recently proposed by \cite{del12a}. We denote it by DHB. Given a classifier, the DHB method proposes a leave-one-out choice of the best variables for the considered classification problem. While this is a worthwhile natural idea, it is computationally intensive. 
So the authors implement a slightly modified version, which we have closely followed. It is based on a sort of trade-off between full and sequential search, together with some additional computational savings. 
Let us note, as an important difference with our maxima-hunting method, that the 
DHB procedure is a ``wrapper'' method, in the sense that it depends on the chosen classifier. Following \cite{del12a}, we have only implemented the DHB method with the LDA classifier.

Apart from that, we proceed as in the simulation study except for the generation of the training, validation and test samples. Here we consider the usual cross-\-validation procedure which avoids splitting the sample (sometimes small) into three different sets. Each output is obtained by standard leave-one-out cross-\-validation. The only exception is the phoneme data set for which this procedure is extremely time-consuming (due to the large sample size); so we use instead ten-fold cross-validation  (10CV). The respective validation steps are done with the same resampling schemes within the training samples. This is a usual way to proceed when working with real data; see \citet[Subsection 7.10]{has09}.
 Several outputs  are given in Tables 2 (accuracy) and 3 (number of variables) below.  The complete  results  can be found in \url{www.uam.es/antonio.cuevas/exp/outputs.xlsx}.

\smallskip

\begin{table}
\caption{\footnotesize  Classification accuracy (in \%) for the real data with both classifiers.}

\begin{center}\footnotesize
    \begin{tabular}{lcccccccc}\hline\noalign{\smallskip}
      \multicolumn{9}{c}{\rm \bf $k$-NN outputs}\\
        Data & FCQ & MID & T & PLS & MHR & MHV & DHB & Base\\
      \hline\noalign{\smallskip}
Growth &        83.87 & \framebox{95.70} & 83.87 & 94.62 & \framebox{95.70} & 94.62 & - & \framebox{96.77} \\
Tecator & 99.07 & 99.07 & 99.07 & 97.21 & \framebox{99.53} & \framebox{99.53} & - & 98.60 \\
Phoneme & \framebox{80.43} & 79.62 & \framebox{80.43} & \framebox{82.53} & 80.20 & 78.86 & - & 78.97 \\
        
      \noalign{\smallskip}\hline  \noalign{\smallskip}\noalign{\smallskip}
      \multicolumn{9}{c}{\rm \bf LDA outputs}\\    
        Data & FCQ & MID & T & PLS & MHR & MHV & DHB &Base\\
      \hline\noalign{\smallskip}
       Growth & 91.40 & 94.62 & 91.40 & 95.70 & 95.70 & \framebox{96.77} & \framebox{96.77} & - \\
Tecator & 94.42 & \framebox{95.81} & 94.42 & 94.42 & \framebox{95.35} & 94.88 & \framebox{95.35} & - \\
Phoneme & 79.38 & \framebox{80.37} & 79.09 & \framebox{80.60} & 80.20 & 78.92 & 77.34 & - \\
      \noalign{\smallskip}\hline
    \end{tabular} 
\end{center}
\end{table}

\smallskip

\begin{table}
	\caption{\footnotesize  Average number of variables (or components) selected for the real data sets.}
\begin{center}\footnotesize
    \begin{tabular}{lcccccccc}\hline\noalign{\smallskip}
      \multicolumn{9}{c}{\rm \bf $k$-NN outputs}\\
        Data & FCQ & MID & T & PLS & MHR & MHV & DHB & Base\\
      \hline\noalign{\smallskip}
Growth &       \framebox{ 1.0} & 3.5 & \framebox{1.0} & 2.8 & 4.0 & 4.0 & - & 31 \\
Tecator & 3.0 & 5.7 & 3.0 & 2.7 & \framebox{1.0} & \framebox{1.0} & - & 100 \\
Phoneme & \framebox{10.7} & 15.3 & 12.3 & 12.9 & \framebox{10.2} & 12.3 & - & 256 \\        
      \noalign{\smallskip}\hline  \noalign{\smallskip}\noalign{\smallskip}
      \multicolumn{9}{c}{\rm \bf LDA outputs}\\    
        Data & FCQ & MID & T & PLS & MHR & MHV & DHB &Base\\
      \hline\noalign{\smallskip}
       Growth & 5.0 & 3.4 & 5.0 & \framebox{2.0} & 4.0 & 4.0 & \framebox{2.3} & - \\
Tecator & 8.4 & 2.6 & 3.1 & 9.7 & \framebox{1.7} & \framebox{1.8} & 3.0 & - \\
Phoneme & 8.5 & 17.1 & \framebox{7.9} & 15.5 & 16.1 & 11.0 & \framebox{2.0} & - \\
      \noalign{\smallskip}\hline
    \end{tabular}
\end{center}
\end{table}

These results are similar to those obtained in the simulation study. While (as expected) there is no clear global winner, maxima-hunting method looks as a very competitive choice. In particular, Tecator outputs are striking, since MHR and MHV achieve (with $k$-NN) a near perfect classification with just one variable.
Note also that maxima-hunting methods (particularly MHR) outperform or are very close to the Base outputs (which uses the entire curves). 
PLS is overcome by our methods in two of the three problems but it is the clear winner in phoneme example. In any case, it should be kept in mind, as a counterpart, the ease of interpretability of the variable selection methods. 

The DHB method performs well in the two first considered examples but relatively fails in the phoneme case. There is maybe some room for improvement in the stopping criterion (recall that we have used the same parameters as in \cite{del12a}). Recall also that, by construction, this is (in the machine learning terminology)  a ``wrapper'' method. This means that the variables selected by DHB are specific for the LDA classifier (and might dramatically change with other classification rules). 
Also note that the use of the LDA classifier didn't lead to any significant gain; in fact, the results are globally worse than those of $k$-NN except for a few particular cases.

Although our methodology is not primarily targeted to the best classification rate, but to the choice of the most representative variables, we can conclude that MH procedures combined with the simple $k$-NN are competitive when compared with PLS and other successful and sophisticated methods in literature: see \citet{gal14} for Tecator data, \citet{mos14} for growth data and \citet{del12a} for phoneme data.

\section{Overall conclusions: a tentative global ranking of methods} \label{conclusiones}
We have summarized the conclusions of our 400 simulation experiments in three  rankings, prepared with different criteria, according to \bf classification accuracy\rm.   With the \bf relative ranking \/ \rm criterion, the winner method (with performance $W$) in each of the 400 experiments gets 10 score points, and the method with the worst performance (say $w$) gets 0 points. The score 
 of any other method, with performance $u$ is just assigned in a proportional way: $10(u-w)/(W-w)$. The \bf positional ranking\/ \rm scoring criterion just gives 10 points to the winner in every experiment, 9 points to the second one, etc.  Finally, the \textbf{F1 ranking} rewards strongly the winner. For each experiment, points are divided as in an F1 Grand Prix: the winner gets 25 points and the rest 18, 15, 10, 8, 6 and 4 successively.
The final average scores are given in Table 4. The winner and the second best methods in each category appear marked.

\begin{table}
	\caption{\footnotesize    Average ranking scores over the 400 experiments.}    
\par
\vskip .2cm
\begin{center}\footnotesize
    \begin{tabular}{lccccccc}\hline\noalign{\smallskip}
      \multicolumn{8}{c}{\rm \bf $k$-NN rankings}\\
        Ranking criterion & FCQ & MID & T & PLS & MHR & MHV &  Base\\
      \hline\noalign{\smallskip}
        Relative  & 4.42 & 5.80 & 2.93 & 6.99 & \framebox{8.42} & \framebox{7.35} & 3.64\\
        Positional  & 6.44 & 6.71 & 5.50 & \framebox{7.96} & \framebox{8.68} & 7.84 & 5.89\\
        F1 & 11.62 & 12.04 & 9.46 & \framebox{17.39} & \framebox{17.96} & 15.41 & 10.15\\
      \noalign{\smallskip}\hline  \noalign{\smallskip}\noalign{\smallskip}
      \multicolumn{8}{c}{\rm \bf LDA rankings}\\    
        Ranking criterion & FCQ & MID & T & PLS & MHR & MHV  & Base\\
      \hline\noalign{\smallskip}
        Relative & 3.76 & 5.19 & 1.96 & 6.90 & \framebox{8.62} & \framebox{8.07} & - \\
        Positional & 6.70 & 6.99 & 5.92 & 8.13 & \framebox{8.79} & \framebox{8.49} & -\\
        F1  & 11.95 & 12.52 & 10.22 & \framebox{17.49} & \framebox{18.41} & 17.47 & -\\
      \noalign{\smallskip}\hline
    \end{tabular}
\end{center}\end{table}

The results are self-explanatory. Nevertheless, \bf the following conclusions might be of some interest for practitioners\rm:

1.  The maxima-hunting methods are the global winners (in particular when using the distance correlation measure), even if there is still room for improvement in the maxima identification.
In fact, the maxima-hunting procedures   result in  accuracy improvements (with respect to the ``base error'', i.e., using the whole trajectories) in 88.00\% of the considered experiments. Overall, the gain of accuracy 
associated with \bf MHR \rm variable selection is relevant (2.41\%).

2. While the univariate ranking methods, such as the $t$ ranking, (which ignore the dependence between the involved variables) are still quite popular among practitioners, they are clearly outperformed by the ``functional'' procedures. It is quite remarkable the superiority of the maxima-hunting methods on the rest of variable selection procedures, requiring often a lesser number of variables.

3. As an important overall conclusion, variable selection appears as a \bf highly competitive alternative to PLS\rm, which is so far the standard dimension reduction method in high-dimensional and functional statistics (whenever a response variable is involved). The results of the above rankings show that variable selection offers a better balance in terms of both accuracy and interpretability.

4. On average, the use of the classical Fisher's discriminant rule LDA (after dimension reduction) provides worse results than the 
nonparametric $k$-NN rule. An example of superiority of a linear classifier is shown in \cite{del12b} where an asymptotic optimality result is provided. In addition, under some conditions, the proposed classifier turns out to be ``near-perfect'' (in the sense that the probability of classification error can be made arbitrarily small) to discriminate between two Gaussian processes.  This is an interesting  phenomenon which does not appear in the finite dimensional case. However, it requires that the Gaussian measures under discrimination are mutually singular (note that this situation cannot happen with two non-degenerate Gaussian measures in ${\mathbb R}^d$). This topic will be considered in a forthcoming manuscript by the authors. 

\smallskip 

\noindent \it A final remark\rm. The present study shows that there are several quite natural models in which the maxima-hunting method is definitely to be recommended. The real data results are also encouraging.  Our results suggest that, even when there is no clear, well-founded guess on the nature of the underlying model, the idea of selecting the maxima of the distance correlation is a suitable choice, that always allows for a direct interpretation. It is natural to ask what type of models would typically be less favorable for the maxima-hunting approach.  As a rough, practical guide, we might say that some adverse situations might typically arise in those cases where the trajectories are extremely smooth,  or when they are very wiggly, with 
many noisy abrupt peaks which tend to mislead the calculation of the maxima in the distance correlation function.

\vskip 14pt
\noindent {\large\bf Supplementary Materials.} All the proofs and two auxiliary results can be found in the appendix. 
Some further methodological and technical details are explained in the Supplementary Materials document below. It also includes  some  extra simulation outputs and the list of the 100 considered models. The full simulation outputs are included in an Excel file downloadable from \url{www.uam.es/antonio.cuevas/exp/outputs.xlsx}.
\par
%%%%%%%%%%%%%%%%%%%%%%%%%%%%%%%%%%%%%%%%%%%%%%%%%%%%%%%%%%%%%%%%%%%%%%%%%%%%%%%%%%%%%%%%%%%%%%%%%%%%%%%%%%%%%%%%%%%%%%%%%%%%
\vskip 14pt
\noindent {\large\bf Acknowledgment.}
This research has been supported by Spanish grant MTM2013-44045-P.

\setcounter{section}{8}
\section*{Appendix: Some results and proofs}\label{sec:prop}
To prove Theorem 2 we need  two lemmas dealing with the uniform strong consistency of one-sample and two-sample functional U-statistics, respectively.

\begin{lemma}
	\label{lemma:Ustatistics}
	Let $X: T\to \mathbb{R}$ be a process with continuous trajectories a.s. defined on the compact rectangle $T=\prod_{i=1}^d [a_i,b_i]\subset \mathbb{R}^d$. Let $X_1,\ldots, X_n$ be a sample of $n$ independent trajectories of $X$. Define the functional U-statistic
	\[
	U_n(t) = \frac{2}{n(n-1)} \sum_{i<j} k[X_i(t), X_j(t)],
	\]
	where the kernel $k$ is a real continuous, permutation symmetric function. Assume that   
	$$
	\mathbb{E}\big(\sup_{t\in T} |k[X(t),X'(t)]|\big)<\infty,
	$$
	where $X$ and $X'$ denote two independent copies of the process. Then, as $n\to\infty$, $\|U_n - U\|_\infty \to 0,\ \ \mbox{a.s.,}$
	where $U(t)=\mathbb{E}(k[X(t),X'(t)])$.
\end{lemma}

\begin{proof}
	First, we show that $U(t)$ is continuous. Let $t_n \subset T$ such that $t_n\to t$. Then, due to the continuity assumptions on the process and the kernel, $k[X(t_n),X'(t_n)]\to k[X(t),X'(t)]$, a.s. Using the assumption $\mathbb{E}\big(\sup_{t\in T} |k[X(t),X'(t)]|\big)<\infty$,  Dominated Convergence Theorem (DCT) allows us to deduce $U(t_n)$ $\to U(t)$.

	Let $M_\delta(t)=\sup_{s: |s-t|_d \leq \delta} |h(s) - h(t)|$ where, for the sake of simplicity, we denote $h(t) = k[X(t),X'(t)]$. The next step is to prove that, as $\delta \downarrow 0$,
	\begin{equation}
	\label{eq.sup}
	\sup_{t\in T} \mathbb{E} ( M_\delta(t)) \to 0.
	\end{equation}
	Both $M_\delta(t)$ and $\lambda_\delta (t)= \mathbb{E}(M_\delta(t))$ are continuous functions. Since $h(t)$ is uniformly continuous on $\{s: |s-t|_d \leq\delta\}$, $M_\delta(t)$ is also continuous. The fact that $\lambda_\delta (t)$ is continuous follows directly from DCT since $|M_\delta(t)| \leq 2\sup_{t\in T} |h(t)|$ and, by assumption, $\mathbb{E}(\sup_{t\in T} |h(t)|)<\infty$. By continuity, $M_\delta (t)\to 0$ and  $\lambda_\delta (t)\to 0$, as $\delta\downarrow 0$. Now, since $\delta > \delta'$ implies $\lambda_\delta(t) \geq \lambda_{\delta'}(t)$, for all $t\in T$, we can apply Dini's Theorem to deduce that $\lambda_\delta(t)$ converges uniformly to 0, that is, $\sup_{t\in T}\lambda_\delta(t)\to 0$, as $\delta\downarrow 0$.

	The last step is to show $\|U_n - U\|_\infty \to 0$ a.s., as $n\to\infty$. For $i\neq j$, denote $M_{ij,\delta}(t)=\sup_{s:|s-t|_d<\delta} |h_{ij}(s) - h_{ij}(t)|$, where $h_{ij}(t) = k[X_i(t),X_j(t)]$, and $\lambda_\delta (t)=\mathbb{E}(M_{ij,\delta}(t))$. Fix $\epsilon > 0$. By (\ref{eq.sup}), there exists $\delta >0$ such that $\lambda_{\delta} (t) < \epsilon$, for all $t\in T$. Now, since $T$ is compact, there exist $t_1,\ldots, t_m$ in $T$ such that $T= \cup_{k=1}^m B_k$, where $B_k = \{t: |t-t_k|_d \leq \delta\}\cap T$. Then,
	\begin{align*}
	\|U_n - U\|_\infty & = \max_{1\leq k \leq m} \sup_{t\in B_k} |U_n(t) - U(t)| \\
	&\leq 
	\max_{1\leq k \leq m} \sup_{t\in B_k}[|U_n(t) - U_n(t_k)| + |U_n(t_k) - U(t_k)| + |U(t_k) - U(t)|  ]\\
	&\leq
	\max_{1\leq k \leq m} \sup_{t\in B_k}|U_n(t) - U_n(t_k)| + \max_{k=1,\ldots,m}|U_n(t_k) - U(t_k)| + \epsilon,\\
	\end{align*}
	since $|s-t|_d \leq \delta$ implies $|U(s)-U(t)| = |\mathbb{E}[h(s) - h(t)] | \leq \mathbb{E}|h(s)-h(t)|\leq \lambda_\delta(t) < \epsilon.$
	
	For the second term, we have $\max_{k=1,\ldots,m}|U_n(t_k) - U(t_k)| \to 0$ a.s., as $n\to\infty$, applying SLLN for U-statistics (see e.g. DasGupta (2008), Theorem 15.3(b), p. 230). As for the first term, observe that using again SLLN for U-statistics,
	\begin{align*}
	\sup_{t\in B_k}|U_n(t) - U_n(t_k)| &\leq  
	\frac{2}{n(n-1)} \sum_{i<j} \sup_{t\in B_k} |h_{ij}(t_k) - h_{ij}(t)| \\ 
	&=  \frac{2}{n(n-1)} \sum_{i<j} M_{ij,\delta}(t_k) \to \lambda_\delta(t_k), \ \ \mbox{a.s.}, 
	\end{align*}
	where $\lambda_\delta(t_k)<\epsilon$. Therefore, 
	\begin{align*}
	\limsup_n\Vert U_n-U\Vert_\infty 
	& \leq  
	\limsup_n\max_{k=1,\ldots, m}\sup_{t\in B_k}|U_n(t) - U_n(t_k)| \\ 
	& +\limsup_n\max_{k=1,\ldots,m}|U_n(t_k) - U(t_k)|+\epsilon  \leq 2\epsilon. \\
	\end{align*}
	
\end{proof}

\begin{lemma}
	\label{lemma:twosampleUstatistics}
	Let $X^{(0)}: T\to \mathbb{R}$ and $X^{(1)}: T\to \mathbb{R}$ be a pair of independent processes with continuous trajectories a.s. defined on the compact rectangle $T=\prod_{i=1}^d [a_i,b_i]$ $\subset \mathbb{R}^d$. Let $X^{(0)}_1,\ldots, X^{(0)}_{n_0}$ and $X^{(1)}_1,\ldots, X^{(1)}_{n_1}$ be  samples of $n_0$ and $n_1$ independent trajectories of $X^{(0)}$ and $X^{(1)}$, respectively. Define the functional two-sample  U-statistic
	\[
	U_{n_0,n_1}(t) = \frac{1}{n_0n_1} \sum_{i=1}^{n_0} \sum_{j=1}^{n_1} k[X^{(0)}_i(t), X^{(1)}_j(t)],
	\]
	where the kernel $k$ is a continuous, permutation symmetric function. Assume that 
	$$
	\mathbb{E}\big(\sup_{t\in T} |h(t)|\log^+|h(t)|\big)<\infty,
	$$
	with $h(t)=k[X^{(0)}(t),X^{(1)}(t)]$. Then, as $\min(n_0, n_1) \to\infty$,
	\[
	\|U_{n_0,n_1} - U\|_\infty \to 0,\ \ \mbox{a.s.,}
	\]
	where $U(t)=\mathbb{E}(k[X^{(0)}(t),X^{(1)}(t)])$.
\end{lemma}

\begin{proof}
	It is  analogous to the proof of Lemma \ref{lemma:Ustatistics} so it is omitted. We need to apply a strong law of large numbers for two-sample U-statistics. This  result can be guaranteed under slightly stronger conditions on the moments of the kernel; see \citet[Th.1]{sen77}. Hence the condition $\mathbb{E}\big(\sup_{t\in T} |h(t)|\log^+ |h(t)|\big) < \infty$ in the statement of the lemma.
\end{proof}

\subsection*{Proofs of the main results}

\begin{proof}[Theorem 1]
	
	\noindent (a) From (2.1), as $X_t$ is $d$-dimensional and $Y$ is one-dimensional, taking into account $c_1=\pi$, we have
	\begin{align*}
	\V^2(X_t,Y) & = \parallel \varphi_{X_t , Y} (u,v) - \varphi_{X_t} (u) \varphi_Y(v)\parallel_w ^2\\
	&= \textstyle \frac{1}{\pi c_d}\int_{\mathbb R}\int_{\mathbb{R}^d} | \varphi_{X_t , Y} (u,v) - \varphi_{X_t} (u) \varphi_Y(v)|^2 \frac{1}{|u|_d^{d+1} v^2}du dv .
	\end{align*}
	Let's analyze the integrand,
	{\small \begin{align*}
		\varphi_{X_t , Y}(u,v) - \varphi_{X_t} (u) \varphi_Y(v) &= \E{e^{i u^\top X_t} e^{ivY}}-\E{e^{iu^\top X_t}}\E{ e^{ivY}}\\
		&=\E{(e^{iu^\top X_t}-\varphi_{X_t}(u))(e^{ivY}-\varphi_{Y}(v))}
		\\&=\E{\E{(e^{iu^\top X_t}-\varphi_{X_t}(u))(e^{ivY}-\varphi_{Y}(v))| X}}\\
		&=\E{(e^{iu^\top X_t}-\varphi_{X_t}(u))\E{(e^{ivY}-\varphi_{Y}(v))| X}}
		\\&\overset{(*)}{=}\E{(e^{iu^\top X_t}-\varphi_{X_t}(u))(e^{iv}-1)(\eta(X)-p)}\\
		&=(e^{iv}-1)\E{(e^{iu^\top X_t}-\varphi_{X_t}(u))(\eta(X)-p)}
		\\&=(e^{iv}-1)\E{e^{iu^\top X_t}(\eta(X)-p)} = (e^{iv}-1)\zeta(u,t).
		\end{align*}}
	Step (*) in the above chain of equalities is motivated as follows:
	\begin{align*}
	\E{(e^{ivY}-\varphi_{Y}(v))| X} &= \E{e^{ivY}| X}-\varphi_{Y}(v)
	=(e^{iv}-1)\eta(X)  - (e^{iv}-1)p \\ &= (e^{iv}-1) ((\eta(X)-p)).
	\end{align*}
	%For step (**) we have used $\E{\eta(X)}  = \E{\PP(Y=1|X)} = \PP(Y=1)=p$.
	Therefore, since
	$\int_{\mathbb R} \frac{|e^{iv}-1|^2}{\pi v^2}dv=2$,
	\begin{align*}
	\V^2(X_t,Y) = \int_{\mathbb R} \frac{|e^{iv}-1|^2}{\pi v^2}dv
	\int_{\mathbb{R}^d} \frac{|\zeta(u,t)|^2}{c_d|u|_d^{d+1}}du 
	= \frac{2}{c_d} \int_{\mathbb{R}^d} \frac{\abs{\zeta(u,t)}^2}{|u|_d^{d+1}}du.
	\end{align*}

	\
	
	\noindent (b) Since $\zeta(u,t)=\E{\left( \eta(X)-p\right)e^{iu^\top X_t}}$,
	\begin{small}
		\begin{align*}
		\abs{\zeta(u,t)}^2 &= \EE \left[ (\eta(X)-p) e^{iu^\top X_t} \right] \EE \left[ (\eta(X')-p) e^{-iu^\top X'_t} \right] \\
		&= \EE \left[ (\eta(X)-p)(\eta(X')-p) e^{iu^\top(X_t-X'_t)} \right]\\  
		&=\EE \left[ (\eta(X)-p)(\eta(X')-p) \cos(u^\top (X_t-X'_t)) \right] \\
		&=- \EE \left[ (\eta(X)-p)(\eta(X')-p)(1- \cos(u^\top(X_t-X'_t))) \right],
		\end{align*}
	\end{small}
	where we have used $\abs{\zeta(u,t)}^2 \in \R$   and 
	$\EE\left[ (\eta(X)-p)(\eta(X')-p)\right]=0$. Now, using expression (3.1),
	\begin{small}
		\begin{align*}
		\V^2(X_t,Y)&= - 2 \EE \left[ (\eta(X)-p)(\eta(X')-p) \int_{\R^d} \frac{1- \cos(u^\top(X_t-X'_t))}{c_d |u|_d^{d+1}} du \right] \\
		&= -2 \EE \left[ (\eta(X)-p)(\eta(X')-p)\abs{X_t - X'_t}_d \right]\\
		&=-2 \EE \left[ (Y-p)(Y'-p)\abs{X_t - X'_t}_d \right],
		\end{align*}
	\end{small}
	since [see e.g. Lemma 1 in \citet{sze07}], 
	$$
	\int_{\mathbb{R}^d}\frac{1-\cos(u^\top x)}{c_d |u|_d^{d+1}}du=|x|_d,\ \ \mbox{for all }x\in {\mathbb{R}^d}.
	$$
	
	\
	
	\noindent (c) By conditioning on $Y$ and $Y'$ we have
	{\small \begin{align*}
		{\mathbb E}[(Y-p)(Y'-p)|X_t - X'_t|_d] &=
		p^2 I_{00}(t) (1-p)^2 - p(1-p)I_{01}(t) 2p(1-p)\\
		\hspace{0.5cm}&\hspace{10pt}+(1-p)^2I_{11}(t)p^2
		=p^2(1-p)^2 (I_{00}(t)+I_{11}(t) - 2I_{01}(t)).
		\end{align*}}
	Now, using (3.2), $ \V^2(X_t,Y) = 4p^2(1-p)^2 \left[ I_{01}(t) - \frac{I_{00}(t)+I_{11}(t)}{2}\right]$.
\end{proof}

\begin{proof}[Theorem 2]
	Continuity of ${\cal V}_n^2(X_t,Y)$ is straightforward from DCT. It suffices to prove the result for sequences of samples $X_1^{(0)},\ldots,X_{n_0}^{(0)}$, and $X_1^{(1)},\ldots,X_{n_1}^{(1)}$, drawn from $X|Y=0$ and $X|Y=1$, respectively, such that $n_1/(n_0+n_1)\to p={\mathbb P}(Y=1)$.

	From the triangle inequality it is enough to prove the	 uniform convergence of $\hat I_{00}(t)$, $\hat I_{11}(t)$ and $\hat I_{01}(t)$ to $I_{00}(t)$, $I_{11}(t)$ and $I_{01}(t)$, respectively. For the first two quantities we apply Lemma \ref{lemma:Ustatistics}  to the kernel $k(x,x')=|x-x'|$. For the last one we apply Lemma \ref{lemma:twosampleUstatistics}  to the same kernel.
	Observe that  $\mathbb{E}\|X\|_\infty < \infty$ implies the moment condition of Lemma \ref{lemma:Ustatistics} whereas $\mathbb{E}( \|X\|_\infty\log^+\|X\|_\infty )< \infty$ implies the moment condition of Lemma \ref{lemma:twosampleUstatistics}.
	The last statement readily follows from the uniform convergence and the 
	compactness of $[0,1]^d$.
\end{proof}

\begin{proof}[Proposition 1]

	We know $g^*(x)={\mathbb I}_{\{\eta(x)>1/2\}}$.  Then, we use equation (4.1), which provides $\eta(x)$
	in terms of the Radon-Nikodym derivative $d\mu_0/d\mu_1$, and the expression for $d\mu_0/d\mu_1$ given in \citet{lip77}, p. 239. This gives
	\[
	\eta(x)=\left[\frac{1-p}{p}\sqrt{2}e^{-x_1^2/4}+1 \right]^{-1}.
	\]		
	
	Now, from $g^*(x)={\mathbb I}_{\{\eta(x)>1/2\}}$, we get  $g^*(x)=1$ if and only if
	$x_1^2 > 4\log\left( \frac{\sqrt{2}(1-p)}{p} \right)$.
\end{proof}
\begin{proof}[Proposition 2]
	
	Again, we use expression (4.1) to derive the expression of the optimal rule $g^*(x)={\mathbb I}_{\{\eta(x)>1/2\}}$. In this case the calculation is made possible using the expression of the Radon-Nikodym derivative for the distribution of a Brownian process with trend, $F(t)+B(t)$,
	with respect to that of a standard Brownian: 
	\begin{equation}
	\frac{d\mu_1}{d
		\mu_0}(B)=\exp\left\{-\frac{1}{2}\int_0^1F^\prime(s)^2ds+\int_0^1F^\prime dB\right\},\label{RNderivative}
	\end{equation}
	for $\mu_0$-almost all $B\in {\mathcal C}[0,1]$; 
	see, \citet{mor10}, Th. 1.38 and Remark 1.43, for further details.
	Observe that in this case 
	we have $F(t)=ct$. Thus, from (4.1), we finally get
	$
	\eta(x) = \left[\frac{1-p}{p}\exp\left(\frac{c^2}{2}-cx_1\right) +1 \right] ^{-1},
	$
	which again only depends on $x$ through $x(1)=x_1$. The result follows easily from  this expression.
\end{proof}
\begin{proof}[Proposition 3]
	
	In this case, the trend function is $F(t)=\Phi_{m,k}(t)$. So $F^{'}(t)=\varphi_{m,k}$ and $F^{''}(t)=0$.
	From equations (4.1) and \eqref{RNderivative},  we readily get (4.3) and (4.4).
\end{proof}

\begin{proof}[Proposition 4]
	Let us first consider the model in Proposition 1 (i.e., Brownian vs. Brownian with a stochastic trend). Such model entails that $X_t | Y=0 \sim  N(0,\sqrt{t})$ and $X_t | Y=1 \sim  N(0,\sqrt{t^2 + t})$. 
	Now, recall that if $\xi\sim N(m,\sigma)$, then,
	\begin{equation}\label{Evalorabsoluto}
	\EE\abs{\xi} =\sigma \sqrt{\frac{2}{\pi}} e^{- \frac{m^2}{\sigma^2}} +
	m \left( 2 \Phi \left(\frac{m}{\sigma}\right)-1\right),
	\end{equation}
	where $\Phi(z)$ denotes the distribution function of the standard normal.
	
	Now, using (3.3) and \eqref{Evalorabsoluto} we have the following expressions,
	$$I_{01}(t) = \EE |\sqrt{t}Z - \sqrt{t^2 + t}Z'|
	=
	\sqrt{\frac{2(t^2 + 2t)}{\pi}},$$
	$$ 
	I_{00}(t) = \EE |\sqrt{t}Z - \sqrt{t}Z'| = \sqrt{\frac{4t}{\pi}},$$
	$$I_{11}(t) = \EE |\sqrt{t^2 + t}Z - \sqrt{t^2 + t}Z'| = \sqrt{\frac{4(t^2+t)}{\pi}},$$
	where $Z$ and $Z^\prime$ are independent $N(0,1)$ random variables.
	
	Then, the function $\V^2(X_t,Y)=4p^2(1-p)^2 \left( I_{01}(t) - \frac{I_{00}(t)+I_{11}(t)}{2}\right)$ grows with $t$ so it is maximized at $t^*=1$, which is the only point that has an influence on
	the Bayes rule.
	
	Let us now consider the model in Proposition 2 (i.e., Brownian vs. Brownian with a linear trend). Again,
	from \eqref{Evalorabsoluto} we have in this case,
	\begin{align*}
	I_{01}(t) = \EE |ct + \sqrt{t}Z - \sqrt{t}Z'| =2\sqrt{\frac{t}{\pi}}e^{-\frac{c^2t}{2}} +ct\left(2\Phi\left(c\sqrt{\frac{t}{2}}\right) -1\right),
	\end{align*}
	$$I_{00}(t)=I_{11}(t) = \EE |\sqrt{t}Z - \sqrt{t}Z'| = \sqrt{\frac{4t}{\pi}},$$
	where $Z$ and $Z^\prime$ are iid standard Gaussian variables. Therefore using (3.3),
	$$\V^2(X_t,Y)= 
	C\left[2\sqrt{\frac{t}{\pi}}\left(e^{-\frac{c^2t}{2}} -1 \right)
	+ct\left(2\Phi\left(c\sqrt{\frac{t}{2}}\right) -1\right)\right],$$
	where $C=4p^2(1-p)^2$.
	We can check  numerically that this an increasing function which reaches its only
	maximum at $t^*=1$. According to Proposition 1 this is the only relevant point for the Bayes rule.
\end{proof}

%%%%%%%%%%%%%%%%%%%%%%%%%%%%%%%%%%%%%%%%%%%%%%%%%%%%%%%%%%%%%%%%%%%%%%%%%%%%%%%%%%%%%%%%%%%%%%%%%%%%%%%%%%%%%%%%%%%%%%%%%%%%

%%%%%%%%%%%%%%%%%%%%%%%%%%%%%%%%%%%%%%%%%%%%%%%%%%%%%%%%%%%%%%%%%%%%%%%%%%%%%%%%%%%%%%%%%%%%%%%%%%%%%%%%%%%%%%%%%%%%%%%%%%%%

\renewcommand\bibname
%%%%%%%%%%%%%%%%%%%%%%%%%%%%%%%%%%%%%%%%%%%%%%%%%%%%%%%%%%%%%%%%%%%%%%%%%%%%%%%%%%%%%%%%%%%%%%%%%%%%%%%%%%%%%%%%%%%%%%%%%%%%

\end{document}